\documentclass[reqno,12pt]{amsart}

\usepackage{amsmath}
\usepackage{latexsym}
\usepackage{amssymb}
\usepackage{hyperref}
\usepackage{graphicx}
\usepackage{epstopdf}
\usepackage{ifpdf}
\ifpdf
\fi
\usepackage{cite}

\makeatletter
 
 \newcommand{\Rmnum}[1]{\expandafter\@slowromancap\romannumeral #1@}
 \makeatother

\newtheorem{theorem}{Theorem}[section]

\newtheorem{proposition}[theorem]{Proposition}

\newtheorem{remark}[theorem]{Remark}

\newcommand{\R}{{\mathbb R}}

\newcommand{\C}{{\mathbb C}}

\newcommand{\be}{\begin{equation}}
\newcommand{\ee}{\end{equation}}
\newcommand{\bea}{\begin{eqnarray}}
\newcommand{\eea}{\end{eqnarray}}
\newcommand{\ba}{\begin{array}}
\newcommand{\ea}{\end{array}}

\newcommand{\ol}{\overline}

\newcommand{\id}{\mathbb{I}}

\newcommand{\re}{\mathrm{Re}}
\newcommand{\im}{\mathrm{Im}}

\newcommand{\eps}{\varepsilon}

\newcommand{\sig}{\sigma}
\newcommand{\Sig}{\Sigma}

\newcommand{\Gam}{\Gamma}
\newcommand{\om}{\omega}

\newcommand{\dta}{\delta}

\newcommand{\tha}{\theta}

\numberwithin{equation}{section}

\begin{document}
\title[RHP for the complex short pulse equation]{Long-time asymptotic behavior for the complex short pulse equation}

\author[J.Xu]{Jian Xu}
\address{College of Science\\
University of Shanghai for Science and Technology\\
Shanghai 200093\\
People's  Republic of China}
\email{jianxu@usst.edu.cn}

\author[E.Fan]{Engui Fan*}
\address{School of Mathematical Sciences, Institute of Mathematics and Key Laboratory of Mathematics for Nonlinear Science\\
Fudan University\\
Shanghai 200433\\
People's  Republic of China}
\email{corresponding author: faneg@fudan.edu.cn}

\date{\today}

\begin{abstract}
In this paper, we consider the initial value problem   for the complex short pulse equation with  a Wadati-Konno-Ichikawa type Lax pair.
 We show that the solution to the initial value problem  has a parametric expression in terms of the solution of $2\times 2$-matrix Riemann-Hilbert problem,
  from which  an   implicit  one-soliton solution
 is obtained   on  the  discrete spectrum.  While on the continuous spectrum  we further 
 establish the  explicit long-time asymptotic behavior of the non-soliton solution by using Deift-Zhou nonlinear steepest descent method.\\[6pt]

\end{abstract}

\maketitle

\section{Introduction}

We  study the long-time asymptotic behavior of solution to the complex short pulse (CSP) equation
\be\label{spe}
u_{xt}+u+\frac{1}{2}(|u|^2u_{x})_{x}=0,
\ee
formulated on the whole line $x\in \R$ with the initial value data
\be\label{spe-ini}
u(x,t=0)=u_0(x)\in \mathcal {S}(\R),
\ee
where $\mathcal{S}(\R)$ denotes the Schwartz space.
\par

The motivation for  investigating  the CSP equation is as follows: In nonlinear optics, it is well known that the nonlinear Schr\"odinger (NLS) equation was always used to model the slowly varying wave trains \cite{yy-1983,hk-1995,a-2001}, which was found that it can be solved by the inverse scattering transform method \cite{zs-nls}. However, when the width of optical pulseis in the order of femtosecond ($10^{-15}$s), the NLS equation becomes less
accurate \cite{r-1992}. Sch\"afer and Wayne
proposed a so-called short pulse (SP) equation
\be\label{sp}
q_{xt}=q+\frac{1}{6}(q^3)_{xx}
\ee
in \cite{swpd}, which provided an increasingly better approximation
to the corresponding solution of the Maxwell equations \cite{cjsw}. Noticing that $q(x,t)$ in equation (\ref{sp}) is a real-valued function, the one-soliton solution (loop soliton) to
the SP equation (\ref{sp}) has no physical interpretation \cite{ssjpa,mjpsj}. Recently, an improvement (\ref{spe}) for the SP equation
was proposed in \cite{fpd-2015}. In contrast with the real-valued function $q(x,t)$ in SP equation (\ref{sp}), $u(x,t)$ in equation (\ref{spe}) is a complex-valued function.  Since
the complex-valued function can contain the information of both amplitude
and phase, it is more appropriate for the description of the optical
waves \cite{yy-1983}. Hence, the equation (\ref{spe}) was so-called complex short pulse equation.
The CSP equation can be viewed as an analogue of the NLS equation
in the ultra-short regime when the width of optical pulse is of the order $10^{-15}$s, see \cite{fpd-2015}. 

\par

It is founded that the CSP equation also admitted a Wadati-Konno-Ichikawa (WKI)-type Lax pair like  the SP equation \cite{fpd-2015, ssjpsj}.
The soliton solutions were obtained by using Hirota method \cite{fpd-2015},   and the multi-breather and higher order rogue
wave solutions to the CSP equation were constructed by the Darboux transformation method in \cite{lfzpd-2016}.

\par


\par
Recently,  in \cite{jiansp} we  obtained  the explicit leading order long-time asymptotic behavior of the solution $q(x,t)$ to the SP equation (\ref{sp})
by using the nonlinear steepest descent method \cite{dz}.  Here we extend above results to  give  the  asymptotic  behavior of  solution  $u(x,t)$  of  the CSP equation,
but it   will    be much   different from  that   on the SP equation (\ref{sp})  in the following three aspects.

{\bf $(i).$} To obtain the Riemann-Hilbert problem corresponding to the initial value problem for the CSP equation, we need an extra transformation (\ref{Transform-d}) when we try to control the behavior of the eigenfunctions when the spectral variable $k\rightarrow \infty$ during our spectral analysis;

 {\bf $(ii).$} In the CSP equation case, there do not exist a symmetry condition
\be\label{sym-sp}
M(x,t,-k)=\sig_2 M(x,t,k)\sig_2
\ee
satisfied by the SP equation;

{\bf $(iii).$} As (\ref{sym-sp}) isn't valid, the order $O(1)$ term of the asymptotic behavior of scale function $\dta(k)$ defined by (\ref{dtadef}) as $k\rightarrow 0$ doesn't equal one, see (\ref{dtalargek}) and (\ref{dta0def}). This fact will affect our final asymptotic formulae.
 {\bf $4).$} We need calculate the model problem around $k=k_0$ and $k=-k_0$ (where $k_0$'s definition is in subsection 5.3), respectively, see subsection (5.3.4). The reason is (\ref{sym-sp}) not valid, again. This leads to there are two phase functions $\tilde \phi_1$ and $\tilde \phi_2$ defined by (\ref{tildephi-12}) which are contained in the asymptotic formulae (\ref{uxtasy-final}), not like the SP equation case.

\par
Organization of this paper is as  follows.  In section 2,  since the associated Lax pair of short pulse (\ref{spe}) has singularities at $k=0$ and $k=\infty$, we perform the spectral analysis to deal with the two singularities, respectively.  In section 3, we formulate the associated Riemann-Hilbert in an alternative space variable $y$ instead of the original space variable $x$.
In this way,  we can reconstruct the solution $u(x,t)$ parameterized from the solution of the Riemann-Hilbert problem via the asymptotic behavior of the spectral variable at $k=0$. 
In section 4, we obtain the one-soliton solution of the CSP equation under the assumption of $a(k)$ having one single zero point. 
In section 5, we obtain the asymptotic relation between $y$ and $x$ through  analyzing the vector Riemann-Hilbert problem with the nonlinear steepest descent method.
finally,  we get  the leading order asymptotic behavior of the solution $u(x,t)$.

\section{Spectral Analysis}
To analyze the long-time asymptotic behavior of the solution of the IVP for the CSP equation on the  line by employing the nonlinear steepest descent method, the first step is to change  the IVP problem into
 a Riemann-Hilbert problem based i=on the fact  that the CSP equation admits a WKI-type Lax pair
\begin{subequations}\label{Laxpair}
\be\label{Lax-x}
\Psi_x=U(x,t,k)\Psi,
\ee
\be\label{Lax-t}
\Psi_t=V(x,t,k)\Psi,
\ee
where
\be\label{Udef}
U=ik U_1=ik(\sig_3+U_{0x}).
\ee
\be\label{Vdef}
V=-\frac{ik}{2}|u|^2 U_1-\frac{1}{4ik}\sig_3+\frac{1}{2}V_0
\ee
with
\be\label{U1def}
U_0=\left(\ba{cc}0&u\\\bar u&0\ea\right),\sig_3=\left(\ba{cc}1&0\\0&-1\ea\right),V_0=\left(\ba{cc}0&u\\-\bar u&0\ea\right).
\ee
\end{subequations}
Here, the $\bar u$ means the conjugate of the complex function $u$.
\par
We can obtain the scattering data by using the $x-$part of Lax pair for analyzing the IVP for the integrable equation via inverse scattering transform method. The $t-$part of Lax pair is only used to determine the time evolution of the scattering data. However,
there are two singularities at $k=\infty$ and $k=0$ in the Lax pair (\ref{Laxpair}).
In order to construct the solution $u(x,t)$ of the CSP equation (\ref{spe}), we need use the expansion of the eigenfunction as spectral parameter $k\rightarrow 0$. This is similar to the short pulse equation \cite{jiansp}, the Camassa-Holm equation \cite{amkst}, and modified Hunter-Saxton equation \cite{amszip}. Hence, in the following we use two different transformations to analyze these two singularities ($k=0$ and $k=\infty$), respectively.

\subsection{For singularity at $k=0$}
$ $
\par

\ \ Introducing the following transformation
\be
\Psi(x,t,k)=\mu^0(x,t,k)e^{(ikx-\frac{t}{4ik})\sig_3},
\ee
then we get the Lax pair of $\mu^0$
\be\label{mu0Laxe}
\left\{
\ba{l}
\mu^0_x-ik[\sig_3,\mu^0]=V_1^0\mu^0,\\
\mu^0_t+\frac{1}{4ik}[\sig_3,\mu^0]=V_2^0\mu^0,
\ea
\right.
\ee
where
\be
V_1^0=ikU_0,\quad V_2^0=-\frac{ik}{2}|u|^2 U_1+\frac{1}{2}V_0.
\ee

Letting $\hat A$ denotes the operators which acts on a $2\times 2$ matrix $X$ by $\hat A X=[A,X]$ , then the Lax pair of $\mu^0$ (\ref{mu0Laxe}) can be written as
\be
d(e^{-(ikx-\frac{t}{4ik})\hat \sig_3}\mu^0)=W^0(x,t,k),
\ee
where $W^0(x,t,k)$ is the closed one-form defined by
\be
W^0(x,t,k)=e^{-(ikx-\frac{t}{4ik})\hat \sig_3}(V_1^0dx+V_2^0dt)\mu^0.
\ee

We define two eigenfunctions $\{\mu^0_j\}_{j=1}^2$ of (\ref{mu0Laxe}) by the Volterra integral equations,
\begin{subequations}\label{mu0jdef}
\be\label{mu01def}
\mu^0_1(x,t,k)=\id+\int_{-\infty}^{x}e^{ik(x-y)\hat\sig_3}V^0_1(y,t,k)\mu^0_1(y,t,k)dy,
\ee
\be\label{mu02def}
\mu^0_2(x,t,k)=\id-\int_{x}^{+\infty}e^{ik(x-y)\hat\sig_3}V^0_1(y,t,k)\mu^0_2(y,t,k)dy.
\ee
\end{subequations}

\begin{proposition}(Analytic property)
From the above definition, we find that the functions $\{\mu^0_j\}_1^2$ are bounded and analytic properties as following:
\begin{itemize}
 \item $[\mu^0_1]_1(x,t,k)$ is bounded and analytic in $D_2$, $[\mu^0_1]_2(x,t,k)$ is in $D_1$;
 \item $[\mu^0_2]_1(x,t,k)$ is bounded and analytic in $D_1$, $[\mu^0_2]_2(x,t,k)$ is in $D_2$.
\end{itemize}
where $[A]_i$ denotes the $i-$th column of a matric $A$, $D_1$ and $D_2$ denote the upper-half and lower-half plane of the complex $k-$sphere, respectively.
\end{proposition}

\begin{proposition}(Asymptotic property)
The functions $\mu^0_j(x,t,k)$ have the expansions in powers of $k$, for $k\rightarrow 0$,
\be\label{mu0asyk0}
\mu^0_j(x,t,k)=\id+\left(\ba{cc}0&iu\\i\bar u&0\ea\right)k+\left(\ba{cc}D^{(2)}_{11}&|u|^2u_x+2u_t\\-|u|^2\bar u_x-2\bar u_t&D^{(2)}_{22}\ea\right)k^2+O(k^3).
\ee
where $D^{(2)}_{11}$ and $D^{(2)}_{22}$ satisfy the following differential equation system,
\be
\left\{
\ba{ll}
D^{(2)}_{11x}=-u_x\bar u,&D^{(2)}_{22x}=-u\bar u_x,\\
D^{(2)}_{11t}=\frac{1}{2}|u|^2(\bar u u_x-u\bar u_x)-u\bar u_t,&D^{(2)}_{22t}=-\frac{1}{2}|u|^2(\bar u u_x-u\bar u_x)-u_t\bar u,
\ea
\right.
\ee
\end{proposition}
\begin{proof}
This asymptotic behavior (\ref{mu0asyk0}) is followed by the expansion of the $\mu^0(x,t,k)$ as $k\rightarrow 0$,
\be
\mu^{0}(x,t,k)=D^{(0)}(x,t)+kD^{(1)}(x,t)+k^2D^{(2)}(x,t)+O(k^3),
\ee
and set the matric $D^{(i)}(x,t)=\left(\ba{cc}D^{(i)}_{11}&D^{(i)}_{12}\\D^{(i)}_{21}&D^{(i)}_{22}\ea\right)$,
then inserting the above expansion into the Lax pair of $\mu^0$ (\ref{mu0Laxe}) to compare the order of $k$.
\end{proof}

\subsection{For the singularity at $k=\infty$}
$ $
\par

\ \ To control the eigenfunctions as $k\rightarrow \infty$, we should make some transformations to the original Lax pair (\ref{Laxpair}). Firstly, we need some denotations to do the next a series of transformations.
\par
Define a $2\times 2$ matrix-value function $G(x,t)$ as
\be\label{Gdef}
G(x,t)=\sqrt{\frac{\sqrt{m}+1}{2\sqrt{m}}}\left(\ba{cc}1&-\frac{\sqrt{m}-1}{\bar u_x}\\ \frac{\sqrt{m}-1}{u_x}&1\ea\right),
\ee
where $m$ is a function of $(x,t)$ defined by
\be\label{mfundef}
m=1+|u_x|^2.
\ee
\begin{remark}
Notice that when $u_x\rightarrow 0$, the nominator $\sqrt{m}-1$ is a high order infinitesimal than denominator $u_x$. So, the matrix function $G(x,t)$ is well-defined.
\end{remark}

Introducing a transformation
\be\label{Laxpair-Phi}
\Psi(x,t,k)=G(x,t)\Phi(x,t,k),
\ee
then we have the Lax pair of $\Psi(x,t,k)$ (\ref{Laxpair}) becomes
\be
\left\{
\ba{l}
\Phi_x=ik\sqrt{m}\sig_3\Phi+U^{\Phi}(x,t,k)\Phi,\\
\Phi_t=-\frac{ik}{2}|u|^2\sqrt{m}\sig_3\Phi+V^{\Phi}(x,t,k)\Phi,
\ea
\right.
\ee
where
\begin{subequations}
\be
\ba{l}
U^{\Phi}(x,t,k)=-G^{-1}G_{x}\\
{}=-\left(\ba{cc}
\frac{u_x\bar u_{xx}-\bar u_x u_{xx}}{4\sqrt{m}(\sqrt{m}+1)}&\frac{(\sqrt{m}-1)u_x\bar u_{xx}-(\sqrt{m}+1)\bar u_{x}u_{xx}}{4m\bar u_x}\\
\frac{(\sqrt{m}+1)u_x\bar u_{xx}-(\sqrt{m}-1)\bar u_{x}u_{xx}}{4m\bar u_x}&-\frac{u_x\bar u_{xx}-\bar u_x u_{xx}}{4\sqrt{m}(\sqrt{m}+1)}
\ea\right),
\ea
\ee
and
\be
\ba{rcl}
V^{\Phi}(x,t,k)&=&\displaystyle{-\frac{1}{4ik}\frac{1}{\sqrt{m}}\sig_3}+\frac{1}{4ik}\frac{1}{\sqrt{m}}\left(\ba{cc}0&u_x\\\bar u_x&0\ea\right)\\
{}&{}&{}
-\frac{1}{4\sqrt{m}}\left(\ba{cc}\bar u u_x-u\bar u_x&-\frac{(\sqrt{m}+1)u \bar u_x+(\sqrt{m}-1)\bar u u_x}{\bar u_x}\\
\frac{(\sqrt{m}-1)u \bar u_x+(\sqrt{m}+1)\bar u u_x}{u_x}&u \bar u_x-\bar u u_x
\ea\right)\\
{}&{}&{}
-\left(\ba{cc}
\frac{u_x\bar u_{xt}-\bar u_x u_{xt}}{4\sqrt{m}(\sqrt{m}+1)}&\frac{(\sqrt{m}-1)u_x\bar u_{xt}-(\sqrt{m}+1)\bar u_{x}u_{xt}}{4m\bar u_x}\\
\frac{(\sqrt{m}+1)u_x\bar u_{xt}-(\sqrt{m}-1)\bar u_{x}u_{xt}}{4m\bar u_x}&-\frac{u_x\bar u_{xt}-\bar u_x u_{xt}}{4\sqrt{m}(\sqrt{m}+1)}
\ea\right).
\ea
\ee
\end{subequations}

Define
\be
p(x,t,k)=x-\int_{x}^{\infty}(\sqrt{m(x',t)}-1)dx'+\frac{t}{4k^2}.
\ee
As we can write the CSP equation (\ref{spe}) into the conservation law form:
\be\label{conslaw}
(\sqrt{m})_t=-\frac{1}{2}(|u|^2\sqrt{m})_x,\quad m=1+u^2_x,
\ee
we get
\be
p_x=\sqrt{m},\quad p_t=-\frac{1}{2}|u|^2\sqrt{m}+\frac{1}{4k^2}.
\ee
\par

Then, define
\be
\Phi(x,t,k)=\tilde \Phi(x,t,k)e^{ikp(x,t,k)\sig_3},
\ee
the Lax pair equation of (\ref{Laxpair-Phi}) becomes
\be\label{Laxpair-tildePhi}
\left\{
\ba{l}
\tilde \Phi_x-ikp_x[\sig_3,\tilde \Phi]=\tilde U(x,t,k)\tilde \Phi,\\
\tilde \Phi_t-ikp_t[\sig_3,\tilde \Phi]=\tilde V(x,t,k)\tilde \Phi,
\ea
\right.
\ee
where
\be
\tilde U(x,t,k)=U^{\Phi}(x,t,k),\quad \tilde V(x,t,k)=V^{\Phi}+\frac{1}{4ik}\sig_3.
\ee

In order to formulate a Riemann-Hilbert problem for the solution of the inverse spectral problem, we seek solutions of the spectral problem which approach the $2\times 2$ identity matrix as $k\rightarrow \infty$. It turns out that solution of equation (\ref{Laxpair-tildePhi}) do not exhibit this property, hence the next step is to transform the equation of $\tilde \Phi$ into an equation with the desired asymptotic behavior.

\begin{remark}
Since the function $u(x,t)$ is a complex-valued, this implies that the diagonal elements of the matrix $\tilde U$ do not equal to zero. Then it leads to the solution of equation (\ref{Laxpair-tildePhi}) do not exhibit the required asymptotic behavior as $k\rightarrow \infty$.
\end{remark}

Let us define $\mu(x,t,k)$ as
\be\label{Transform-d}
\Phi(x,t,k)=e^{d_-\hat \sig_3}\mu(x,t,k)e^{-d_+\sig_3}e^{ikp(x,t,k)\sig_3},
\ee
where
\be\label{ddef}
\ba{l}
\displaystyle{d_{-}=\int_{-\infty}^{x}\frac{\bar u_x u_{xx}-u_x \bar u_{xx}}{4\sqrt{m}(\sqrt{m}+1)}(x',t)dx',}\vspace{2mm}\\
\displaystyle{d_{+}=\int_{x}^{+\infty}\frac{\bar u_x u_{xx}-u_x \bar u_{xx}}{4\sqrt{m}(\sqrt{m}+1)}(x',t)dx',}\vspace{2mm}\\
\displaystyle{d=d_{-}+d_{+}=\int_{-\infty}^{+\infty}\frac{\bar u_x u_{xx}-u_x \bar u_{xx}}{4\sqrt{m}(\sqrt{m}+1)}(x',t)dx'.}
\ea
\ee
\begin{remark}
Note that $d$ is a pure image quantity conserved under the dynamics governed by (\ref{spe})
\end{remark}

Then the Lax pair of $\mu(x,t,k)$
can be written as
\be\label{mudiffform}
d(e^{-ikp(x,t,k)\hat\sig_3}\mu)=W(x,t,k),
\ee
where $W(x,t,k)$ is the closed one-form defined by
\be\label{Wdef}
W=e^{-ikp(x,t,k)\hat\sig_3}V(x,t,k)\mu(x,t,k),
\ee
where
\be
V(x,t,k)=e^{-d_-\hat\sig_3}(V_1dx+V_2dt),
\ee
with
\begin{subequations}
\be
V_1=-\left(\ba{cc}
0&\frac{(\sqrt{m}-1)u_x\bar u_{xx}-(\sqrt{m}+1)\bar u_{x}u_{xx}}{4m\bar u_x}\\
\frac{(\sqrt{m}+1)u_x\bar u_{xx}-(\sqrt{m}-1)\bar u_{x}u_{xx}}{4m\bar u_x}&0
\ea\right),
\ee
\be
\ba{rcl}
V_2&=&-\frac{1}{4ik}(\frac{1}{\sqrt{m}}-1)\sig_3+\frac{1}{4ik}\frac{1}{\sqrt{m}}\left(\ba{cc}0&u_x\\\bar u_x&0\ea\right)\\
{}&{}&{}
-\frac{1}{4\sqrt{m}}\left(\ba{cc}0&-\frac{(\sqrt{m}+1)u \bar u_x+(\sqrt{m}-1)\bar u u_x}{\bar u_x}\\
\frac{(\sqrt{m}-1)u \bar u_x+(\sqrt{m}+1)\bar u u_x}{u_x}&0
\ea\right)\\
{}&{}&{}
-\left(\ba{cc}
0&\frac{(\sqrt{m}-1)u_x\bar u_{xt}-(\sqrt{m}+1)\bar u_{x}u_{xt}}{4m\bar u_x}\\
\frac{(\sqrt{m}+1)u_x\bar u_{xt}-(\sqrt{m}-1)\bar u_{x}u_{xt}}{4m\bar u_x}&0
\ea\right).
\ea
\ee

\end{subequations}

We define two eigenfunctions $\{\mu_j\}_1^2$ of (\ref{mudiffform}) by the Volterra integral equations
\begin{subequations}\label{mujdef}
\be\label{mu1def}
\mu_1(x,t,k)=\id+\int_{-\infty}^{x}e^{ik[p(x,t,k)-p(y,t,k)]\hat\sig_3}V_1(y,t,k)\mu_1(y,t,k)dy,
\ee
\be\label{mu2def}
\mu_2(x,t,k)=\id-\int_{x}^{+\infty}e^{ik[p(x,t,k)-p(y,t,k)]\hat\sig_3}V_1(y,t,k)\mu_2(y,t,k)dy.
\ee
\end{subequations}

\begin{proposition}(Analytic property)\label{analycond}
From the above definition, we find that the functions $\{\mu_j\}_1^2$ are bounded and analytic properties as following:
\begin{itemize}
 \item  $[\mu_1]_1(x,t,k)$ is bounded and analytic in $D_2$, $[\mu_1]_2(x,t,k)$ is in $D_1$;
 \item $[\mu_2]_1(x,t,k)$ is bounded and analytic in $D_1$, $[\mu_2]_1(x,t,k)$ is in $D_2$.
\end{itemize}
\end{proposition}

\begin{proposition}(Symmetry property)\label{symcond}
The eigenfunctions $\mu_j(x,t,k), j=1,2$ satisfy the following symmetry condition,
\be
\ol{\mu_j(x,t,\bar k)}=\sig_2\mu_j(x,t,k)\sig_2.
\ee
\end{proposition}

\begin{proposition}(Large $k$ property)
The matrix functions $\mu_j(x,t,k)$ also satisfy the asymptotic condition
\be\label{Masy}
\mu_j(x,t,k)=\id+O(\frac{1}{k}),\quad k\rightarrow \infty,
\ee
where $\id$ is an $2\times 2$ identity matrix.
\end{proposition}

\subsection{The scattering matrix $S(k)$}
$ $
\par

\ \ Because the eigenfunctions $\mu_1(x,t,k)$ and $\mu_2(x,t,k)$ are both the solutions of equation (\ref{mudiffform}), they are related by a matrix $S(k)$ which is independent of the variable $(x,t)$.
\be\label{scatermatrix}
\mu_1(x,t,k)=\mu_2(x,t,k)e^{ikp(x,t,k)\hat\sig_3}S(k).
\ee
By the definition of $\mu_j(x,t,k),j=1,2$ (\ref{mujdef}) and the symmetry property (see, Proposition \ref{symcond}), the matrix $S(k)$ has the form
\be\label{Skdef}
S(k)=\left(\ba{cc}\ol{a(\bar k)}&b(k)\\-\ol{b(\bar k)}&a(k)\ea\right).
\ee
The function $a(k)$ can be computed by
\be\label{akdef}
a(k)=\det{([\mu_2]_1,[\mu_1]_2)},
\ee
where $\det{(A)}$ means the determinate of a matrix $A$. We can deduce that $a(k)$ is analytic in $D_1$ from the analytic property (see, Proposition \ref{analycond}).
\begin{proposition}
The proposition (\ref{promujrelmu0j}) together with (\ref{akdef}) allows expressing the expansions in powers of $k$ of $a(k)$ at $k=0$,
\be\label{akasy}
a(k)=e^{d}(1+ikc-\frac{c^2}{2}k^2+O(k^3)),\quad k\rightarrow 0.
\ee
\end{proposition}

\subsection{The relation between $\mu_j(x,t,k)$ and $\mu^0_j(x,t,k)$}

$ $
\par

\ \
We use the eigenfunctions $\mu_j$ to define the matrix $M(x,t,k)$ (see (\ref{Mdef})) which is used to formulate a Riemann-Hilbert problem. However, in order to construct the solution $u(x,t)$ from the associate Riemann-Hilbert problem, we need the asymptotic behavior of $\mu_j$ as $k\rightarrow 0$. So, we need relate the eigenfunctions $\mu_j(x,t,k)$ to $\mu^0_j(x,t,k)$.
\par
Note that the eigenfunctions $\mu(x,t,k)$ and $\mu^0(x,t,k)$ being related to the same Lax pair (\ref{Laxpair}), must be related to each other as
\be\label{mujrelmu0j}
\mu_j(x,t,k)=e^{-d_-\sig_3}G^{-1}(x,t)\mu^0_j(x,t,k)e^{(ikx-\frac{t}{4ik})\sig_3}C_j(k)e^{-ikp(x,t,k)\sig_3}e^{d\sig_3},
\ee
with $C_j(k)$ independent of $x$ and $t$. Evaluating (\ref{mujrelmu0j}) as $x\rightarrow \pm\infty$ gives
\be
C_1(k)=e^{-d\sig_3}e^{-ikc\sig_3},\quad \quad C_2(k)=\id,
\ee
where $c=\int_{-\infty}^{+\infty}(\sqrt{m(x,t)}-1)dx$ is a quantity conserved under the dynamics governed by (\ref{spe}).

\begin{proposition}\label{promujrelmu0j}
The functions $\mu_j(x,t,k)$ and $\mu^0_j(x,t,k)$ are related as follows:
\begin{subequations}
\be
\mu_1(x,t,k)=e^{-d_-\sig_3}G^{-1}(x,t)\mu^0_1(x,t,k)e^{-ik\int_{-\infty}^x(\sqrt{m(x',t)}-1)dx'\sig_3},
\ee
\be
\mu_2(x,t,k)=e^{-d_-\sig_3}G^{-1}(x,t)\mu^0_2(x,t,k)e^{ik\int_{x}^{+\infty}(\sqrt{m(x',t)}-1)dx'\sig_3}e^{d\sig_3}.
\ee
\end{subequations}
\end{proposition}

\section{The Riemann-Hilbert problem for CSP equation}

Let us define
\be\label{Mdef}
M(x,t,k)=\left\{\ba{cc}\left(\ba{cc}[\mu_2]_1&\displaystyle{\frac{[\mu_1]_2}{a(k)}}\ea\right),&k\in D_1,\\[9pt]
\left(\ba{cc}\displaystyle{\frac{[\mu_2]_1}{\ol{a(\bar k)}}}&[\mu_1]_2\ea\right),&k\in D_2.
\ea
\right.
\ee
From the definition (\ref{Mdef}) and (\ref{mujdef}),  we can deduce that   $M(x,t,k)$  admits  the symmetry 
\be\label{msymcon}
  \ol{M(x,t,\bar k)}=\sig_2M(x,t,k)\sig_2.
  \ee
and satisfies  the following Riemann-Hilbert problem:
\begin{itemize}
  \item Jump condition: The two limiting values
                        \be
                        M_{\pm}(x,t,k)=\lim_{\eps\rightarrow 0}M_{\pm}(x,t,k\pm i\eps),\quad k\in \R,
                        \ee
                        are related by
\be\label{Mjump}
M_+(x,t,k)=M_-(x,t,k)J(x,t,k),\quad k\in \R,
\ee
where 
\be\label{Jdef}
J(x,t,k)=e^{ikp(x,t,k)\hat \sig_3}J_0(k),
\ee
\be
J_0(k)=\left(\ba{cc}
1&r(k)\\
\ol{r(k)}&1+|r(k)|^2
\ea
\right),\ \ r(k)=\frac{b(k)}{a(k)}.
\ee
  \item Normalize condition as $k\rightarrow \infty$
\be
M(x,t,k)=\id+O(\frac{1}{k}).
\ee
\end{itemize}
In order to get the information of the solution $u(x,t)$, we should consider the asymptotic behavior of $M(x,t,k)$ as $k\rightarrow 0$, that is,
{\small
  \be\label{Masyk0}
  M(x,t,k)=e^{-d_-\sig_3}G^{-1}(x,t)\left[\id+k(ic_+\sig_3+i\left(\ba{cc}0&u\\\bar u&0\ea\right))+O(k^2)\right]e^{d\sig_3},
  \ee
}
  where
  \be
  c_+=\int_x^{+\infty}(\sqrt{m(x',t)}-1)dx'.
  \ee

Equations (\ref{Masyk0}) show that the matrix-valued function $M(x,t,k)$ contains all necessary information for reconstructing the solution of the initial value problem of (\ref{spe})-(\ref{spe-ini}) in terms of the solution of a matrix-valued Riemann-Hilbert problem.

However, the jump relation (\ref{Jdef}) cannot
be used immediately for recovering the solution of CSP equation (\ref{spe})-(\ref{spe-ini}). Since, in the representation of the jump matrix
$e^{ikp(x,t,k)\hat \sig_3}J_0(k)$ the factor $J_0(k)$ is indeed given in terms of the known initial
data $u_0(x)$ but $p(x,t,k)$ is not, it involves $m(x,t)$ which is unknown.
\par
To overcome this, we introduce the new (time-dependent) scale
\be\label{ydef}
y(x,t)=x-\int_{x}^{+\infty}(\sqrt{m(x',t)}-1)dx'=x-c_+(x,t),
\ee
 which make  the jump matrix  explicit,  however,  the solution of the initial problem
can be given only in implicit form: it will be given in terms of
functions in the new scale, whereas the original scale will also be given in terms
of functions in the new scale.

By the definition of the new scale $y(x,t)$, we define
\be\label{muydef}
\tilde M(y,t,k)=M(x(y,t),t,k),
\ee
then we can obtain the Riemann-Hilbert problem of $\tilde M(y,t,k)$ as follows:
  \begin{itemize}
  \item Analyticity: $\tilde M(y,t,k)$ is analytic in the two open half-planes $D_1$ and $D_2$, and continuous up to the boundary $k\in\R$.
  \item Jump condition: The two limiting values
                        \begin{subequations}\label{muyrhp}
                        \be
                        \tilde M_{+}(y,t,k)=\tilde M_{-}(y,t,k)\tilde J(y,t,k),\quad k\in\R,
                        \ee
                        where the jump matrix is
                        \be
                        \tilde J(y,t,k)=e^{i(ky+\frac{t}{4k})\hat\sig_3}J_0(k)
                        \ee
                        with
                        \be\label{J0def}
                        J_0(k)=\left(\ba{cc}
                                     1&r(k)\\
                                     \ol{r(k)}&1+|r(k)|^2
                                     \ea
                                     \right)
                        \ee
                        \end{subequations}
  \item Normalization:
                       \be\label{rhpnorcond}
                       \tilde M(y,t,k)\rightarrow \id,\quad k\rightarrow \infty.
                       \ee
  \end{itemize}

\begin{theorem}\label{rhpthem}
Let $\tilde M(y,t,k)$ satisfies the above conditions, then this Riemann-Hilbert problem has a unique solution. And the solution $u(x,t)$ of the initial value problem (\ref{spe})-(\ref{spe-ini}) can be expressed, in parametric form, in terms of the solution of this Rieamnn-Hilbert problem:
\begin{subequations}\label{cspesolu}
\be
u(x,t)=u(y(x,t),t),
\ee
where
\be
x(y,t)=y+\lim_{k\rightarrow 0}\frac{\left((\tilde M(y,t,0))^{-1}\tilde M(y,t,k)\right)_{11}-1}{ik}
\ee
\be
e^{-2d}u(y,t)=\lim_{k\rightarrow 0}\frac{\left((\tilde M(y,t,0))^{-1}\tilde M(y,t,k)\right)_{12}}{ik}
\ee
\end{subequations}

\end{theorem}

\begin{proof}
 Since the jump matrix $\tilde J(y,t,k)$ is a Hermitian matrix, then the Riemann-Hilbert problem of $\tilde M(y,t,k)$ indeed has a solution. Furthermore, the Riemann-Hilbert problem has only one solution because of the normalize condition.

 \par
 The statements of the solution $u(x,t)$ is following from the asymptotic formula (\ref{Masyk0}).

\end{proof}

\section{Soliton solutions of the CSP equation}

In this section,   we  construct the soliton solutions of the CSP equation.  We should  first address the residue conditions of the Riemann-Hilbert obtained.

\subsection{Residue conditions}
$ $
\par
We recall $a(k)$ is analytic in $D_1$, hence we assume that $a(k)$ has $N$ simple zeros $\{k_j\}_{j=1}^{N}$ in $D_1$.

\par
From the definition of the function $a(k)$ (\ref{akdef}),
we know that if $a(k_j)=0$, then $[\mu_2]_1(x,t,k_j)$ and $[\mu_1]_2(x,t,k_j)$ are linearly
dependent vectors. Then we conclude that there exists a constant $b_j$ such that
\be
[\mu_1]_2(x,t,k_j)=b_je^{2ik_j(y+\frac{t}{4k^2_j})}[\mu_2]_1(x,t,k_j).
\ee

It implies that
\be\label{rhpM2res}
\ba{l}
\mbox{Res}_{k=k_j}[M]_2(x,t,k)=\mbox{Res}_{k=k_j}\frac{[\mu_1]_2(x,t,k)}{a(k)}=\frac{[\mu_1]_2(x,t,k_j)}{\dot{a}(k_j)}\\
{}=C_je^{2ik_j(y+\frac{t}{4k^2_j})}[\mu_2]_1(x,t,k_j)=C_je^{2ik_j(y+\frac{t}{4k^2_j})}[M]_1(x,t,k_j),
\ea
\ee
where $\dot{a}(k)=\frac{d(a(k))}{dk}$ and $C_j=\frac{b_j}{\dot{a}(k_j)}$.
And recall the symmetry condition, the complex conjugate of (\ref{rhpM2res}) is
\be\label{rhpM1res}
\mbox{Res}_{k=\bar k_j}[M]_1(x,t,k)=\bar C_je^{-2i\bar k_j(y+\frac{t}{4\bar k^2_j})}[M]_2(x,t,\bar k_j).
\ee

\subsection{Solitons}
$ $
\par
The solitons correspond to $b(k)$ vanishing identically. In this case the jump matrix $J$ in (\ref{muyrhp}) is the identity matrix and the Riemann¨CHilbert problem of theorem \ref{rhpthem} consists of finding a meromorphic function $M(y,t,k)$ satisfying (\ref{rhpnorcond}) as well
as the residue conditions (\ref{rhpM2res}) and (\ref{rhpM2res}).
From (\ref{rhpnorcond}) and (\ref{rhpM2res}), we get
\be\label{M2soli}
[M(x,t,k)]_2=\left(\ba{c}0\\1\ea\right)+\sum_{j=1}^{N}\frac{C_je^{2ik_j(y+\frac{t}{4k^2_j})}}{k-k_j}[M(x,t,k_j)]_1.
\ee
If we impose the symmetry condition (\ref{symcond}), equation (\ref{M2soli}) can be written as
\be
\left(
\ba{c}
-\ol{M_{21}(x,t,\bar k)}\\
\ol{M_{11}(x,t,\bar k)}
\ea
\right)=\left(\ba{c}0\\1\ea\right)+\sum_{j=1}^{N}\frac{C_je^{2ik_j(y+\frac{t}{4k^2_j})}}{k-k_j}
\left(
\ba{c}
M_{11}(x,t,k_j)\\
M_{21}(x,t,k_j)
\ea
\right).
\ee
Evaluation at $\bar k_j$ yields
\be\label{M1system}
\left(
\ba{c}
-\ol{M_{21}(x,t,k_j)}\\
\ol{M_{11}(x,t,k_j)}
\ea
\right)=\left(\ba{c}0\\1\ea\right)+\sum_{j=1}^{N}\frac{C_je^{2ik_j(y+\frac{t}{4k^2_j})}}{\bar k_j-k_j}
\left(
\ba{c}
M_{11}(x,t,k_j)\\
M_{21}(x,t,k_j)
\ea
\right).
\ee
Solving this algebraic system for $M_{11}(x,t,k_j)$ and $M_{21}(x,t,k_j)$, $j=1,2,\dots N$, and substituting
the solution into (\ref{M2soli}) yields an explicit expression for $[M(x,t,k)]_2$. Then by the symmetry (\ref{symcond}), we deduce $[M(x,t,k)]_1$. This solves the
Riemann-Hilbert problem. Hence, the solution $u(x,t)$ can be expressed perimetrically in terms of the solution of the Riemann-Hilbert problem by (\ref{cspesolu}).

\subsection{One-soliton solution}
$ $
\par

In this subsection we derive explicit formulae for the one-soliton solutions. Assume $N=1$ so that there is one simple zero of $a(k)$, $k_1$ in upper-half plane. We find the algebraic system (\ref{M1system}) reduces to the following two equations
\be
\ba{l}
-\ol{M_{21}(x,t,k_1)}=C_1e^{2ik_1(y+\frac{t}{4k^2_1})}\frac{1}{\bar k_1-k_1}M_{11}(x,t,k_1)\\
\ol{M_{11}(x,t,k_1)}=1+C_1e^{2ik_1(y+\frac{t}{4k^2_1})}\frac{1}{\bar k_1-k_1}M_{21}(x,t,k_1)
\ea
\ee
Solving for $M_{11}(x,t,k_1)$ and $M_{21}(x,t,k_1)$, we find
\begin{subequations}
\be
M_{11}(x,t,k_1)=\frac{1}{1+\frac{|C_1|^2}{4b^2}e^{-4\psi_2}},
\ee
\be
M_{21}(x,t,k_1)=\frac{-\bar C_1e^{-2i\psi_1}}{2ib}\frac{e^{-2\psi_2}}{1+\frac{|C_1|^2}{4b^2}e^{-4\psi_2}},
\ee
where $b$ denotes the image part of $k_1$, i.e. we set $k_1=a+ib$, and $\psi_j=\psi_j(y,t),j=1,2$ denote the real and image part of the function $\psi(y,t)$ defined as follows, respectively
\be
\psi(y,t)=k_1(y+\frac{t}{4k^2_1})=\psi_1+i\psi_2.
\ee
\end{subequations}
A direct calculation shows that
\be
\psi_1=ay+\frac{at}{4(a^2+b^2)},\quad \psi_2=by-\frac{bt}{4(a^2+b^2)}.
\ee
\par
Thus, from the symmetry we get the explicit expression for $M(x,t,k)$
\be
M(x,t,k)=\left(\ba{cc}
1+\frac{1}{k-\bar k_1}\bar m_{22}&\frac{1}{k-k_1}m_{12}\\
-\frac{1}{k-\bar k_1}\bar m_{12}&1+\frac{1}{k-k_1}m_{22},
\ea\right)
\ee
where
\be
m_{12}=C_1e^{2i\psi_1}\frac{e^{-2\psi_2}}{1+\frac{|C_1|^2}{4b^2}e^{-4\psi_2}},\quad m_{22}=-\frac{|C_1|^2}{2ib}\frac{e^{-4\psi_2}}{1+\frac{|C_1|^2}{4b^2}e^{-4\psi_2}}.
\ee

If assuming that $C_1=2be^{2ic+2y_0}$ where $c$ and $y_0$ are real constants, we can find
\be
m_{12}=be^{2i\psi_1+2ic}\frac{1}{\cosh{(2\psi_2+2y_0)}},\quad m_{22}=ib(\tanh{(2\psi_2+2y_0)}-1).
\ee
\par
Hence, we obtain the one-soliton solution as follows,
\be\label{1sln-uyt}
u(x,t)=u(y(x,t),t)=\frac{b}{a^2+b^2}e^{2ic+2d-2i\tha+i\frac{\pi}{2}}e^{2i\psi_1}\frac{1}{\cosh{(2\psi_2+2y_0)}},
\ee
here $\tan{(\tha)}=\frac{b}{a}$, $d$ defined as (\ref{ddef}),
and the relation between variable $x$ and $y$ is
\be\label{1sln-xy}
x=y-\frac{b}{a^2+b^2}(\tanh{(2\psi_2+2y_0)}-1).
\ee

\begin{remark}
The expression (\ref{1sln-uyt}) of the solution represents an envelope soliton of amplitude $\frac{b}{a^2+b^2}$, the velocity $\frac{1}{4(a^2+b^2)}$ and phase $2\psi_1$ in $yt-$coordinate. And from the equation (\ref{1sln-xy}), we have
\be
\frac{\partial x}{\partial y}=1-\frac{2b^2}{a^2+b^2}\frac{1}{\cosh^2(2\psi_2+2y_0)},
\ee
Therefore, $\frac{\partial x}{\partial y}\rightarrow 1$ as $y\rightarrow \pm\infty$.
\end{remark}

\begin{remark}
We should point out that if we set $b=\frac{p_{1R}}{2}$ and $a=\frac{p_{1I}}{2}$, then the solution (\ref{1sln-uyt}) obtained in this paper is equivalent to the expressions in \cite{fpd-2015}. Both of the expressions of the one-soliton solution have the same amplitude and phase. Hence, we conclude that
\begin{itemize}
  \item When $|a|>|b|$, we get smooth soliton;
  \item When $|a|<|b|$, we get loop soliton;
  \item When $|a|=|b|$, we get cuspon soliton.
\end{itemize}
\end{remark}


\section{Long-time asymptotic analysis}

The most important advantage of formulating a Riemann-Hilbert problem to the IVP for the CSP equation is we can rigorously analyze 
the asymptotic behavior of the solution as $t\rightarrow \infty$ via the nonlinear steepest descent method \cite{dz}. 
\par
In this section, we  investigate   the long-time asymptotic behavior of the non-soliton solution of the IVP for the CSP equation on the continuous spectrum. 
To make our analysis be  simple,   we assume that $a(k)$ has no zeros.

\subsection{Main results}

\begin{theorem}\label{main1}
Let $u_0(x)$ satisfy the initial value (\ref{spe-ini}) and be such that no discrete spectrum is present. For $\xi=\frac{x}{t}<-\eps$, $\eps$ be any small positive number, the solution $u(x,t)$ of the initial value problem (\ref{spe})-(\ref{spe-ini}) tends to $0$ fast decay as $t\rightarrow \infty$.
\end{theorem}

\begin{theorem}\label{main2}
Let $u_0(x)$ satisfy the hypotheses of Theorem \ref{main1}. For $\xi=\frac{x}{t}>\eps$, $\eps$ be any small positive number, the solution $u(x,t)$ of the initial value problem (\ref{spe})-(\ref{spe-ini}) equals
\be\label{uxtasy-final}
u(x,t)=\frac{1}{\sqrt{t}}\left[
\sqrt{\frac{-\nu(-\kappa_0)}{\kappa_0}}e^{i\tilde \phi_1}-\sqrt{\frac{-\nu(\kappa_0)}{\kappa_0}}e^{-i\tilde \phi_2}
\right],
\ee
where
\be
\kappa_0=\frac{1}{2}\sqrt{\frac{t}{|x|}},
\ee
\begin{subequations}\label{tildephi-12}
\be
\ba{l}
\tilde \phi_1=-\frac{i\pi}{4}+\mbox{arg}(r(-\kappa_0))+\mbox{arg}(\Gam(i\nu(-\kappa_0)))+\nu(-\kappa_0)\ln(\frac{\kappa_0^3}{t})-\nu(\kappa_0)\ln(4\kappa_0^2)\\
{}-\frac{1}{\pi}\int_{-\kappa_0}^{\kappa_0}\ln(s+\kappa_0)d\ln(1+|r(s)|^2)-\frac{t}{\kappa_0}-\frac{1}{\pi}\int_{-\kappa_0}^{\kappa_0}\frac{\ln(1+|r(s)|^2)}{s}ds-2id-2\kappa_0\dta_1,
\ea
\ee
\be
\ba{l}
\tilde \phi_2=\frac{3i\pi}{4}-\mbox{arg}(r(-\kappa_0))-\mbox{arg}(\Gam(-i\nu(\kappa_0)))+\nu(\kappa_0)\ln(\frac{\kappa_0^3}{t})-\nu(-\kappa_0)\ln(4\kappa_0^2)\\
{}+\frac{1}{\pi}\int_{-\kappa_0}^{\kappa_0}\ln(\kappa_0-s)d\ln(1+|r(s)|^2)-\frac{t}{\kappa_0}-\frac{1}{\pi}\int_{-\kappa_0}^{\kappa_0}\frac{\ln(1+|r(s)|^2)}{s}ds+2id+2\kappa_0\dta_1,
\ea
\ee
\end{subequations}
with $\nu(k)$ defined by (\ref{nudef}), $d$ (which is a pure-image constant) defined by (\ref{ddef}) and $\Gam(x)$ defined as a Gamma function.
\end{theorem}

\begin{remark}
The sectors of different asymptotic behavior match, as $\eps \rightarrow 0$, through the fast decay. Indeed, as $\frac{x}{t}\rightarrow 0^-$, then $\kappa_0\rightarrow \infty$ and $\nu(\kappa_0)\rightarrow 0$ and thus the amplitude in (\ref{uxtasy-final}) decays faster.
\end{remark}

\subsection{Proof of Theorem \ref{main1}}
$ $
\par
The key feature of the method is the deformation of the original Riemann-Hilbert problem according to the signature table for the phase function $\theta$ in jump matrix $\tilde J$ written in the form
\be
\tilde J(y,t,k)=e^{it\tha(\tilde\xi,k)\hat\sig_3}J_0(k),
\ee
where
\be
\tha(\tilde\xi,k)=\tilde\xi k+\frac{1}{4k},
\ee
\be
\tilde\xi=\frac{y}{t}.
\ee
\par
The signature table is the distribution of signs of $\im \tha(\tilde\xi,k)$ in the $k-$plane, 
\be
\im \tha(\tilde\xi,k)=k_2[\tilde\xi-\frac{1}{4(k_1^2+k_2^2)}],
\ee
where $k_1$ and $k_2$ are the real and image part of $k$, respectively, i.e. $k=k_1+ik_2$.
\par
Under the condition $\tilde \xi<-\eps$ for any $\eps>0$, then we have $\im \tha(\tilde\xi,k)>0$ and $\im \tha(\tilde\xi,k)<0$, as $k_2=\im k<0$ and $k_2=\im k>0$, respectively, see figure \ref{fig1}.

\begin{figure}[th]
\centering
\includegraphics{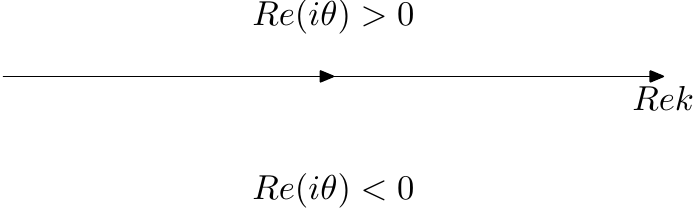}
\caption{\small The signs of $\re i\tha$ in the $k-$plane in the case $\tilde\xi <0$.}\label{fig1}
\end{figure}

This suggests the use of the following factorization of the jump matrix for all $k\in \R$:
\be\label{Jycase1fac}
\tilde J(y,t,k)=\left(\ba{cc}1&\frac{r(k)}{1+|r(k)|^2}e^{2it\tha}\\0&1\ea\right)\left(\ba{cc}\frac{1}{1+|r(k)|^2}&0\\0&1+|r(k)|^2\ea\right)
\left(\ba{cc}1&0\\\frac{\ol{r(k)}}{1+|r(k)|^2}e^{-2it\tha}&1\ea\right)
\ee
\par

In order to deform the original Riemann-Hilbert problem (\ref{muyrhp}), we need introduce a scalar function $\tilde \dta(k)$ which is defined by the following scalar Riemann-Hilbert problem
\be\label{case1dtadef}
\left\{
\ba{ll}
\tilde \dta_+(k)=\tilde \dta_-(k)(1+|r(k)|^2),&k\in \R\backslash {0},\\
\tilde \dta(k)\rightarrow 1,&k\rightarrow \infty,
\ea
\right.
\ee
which  has a explicit solution 
\be
\tilde \dta(k)=\exp\left[\frac{1}{2\pi i}\int_{-\infty}^{\infty}\frac{\ln(1+|r(s)|^2)}{s-k}ds\right].
\ee

Then, we make a transformation of $\tilde M(y,t,k)$ as
\be
\tilde M^{(1)}(y,t,k)=\tilde M(y,t,k)\tilde \dta^{\sig_3}.
\ee
It yields that the Riemann-Hilbert problem for $\tilde M^{(1)}(y,t,k)$ is
\be
\left\{
\ba{ll}
\tilde M^{(1)}_+(y,t,k)=\tilde M^{(1)}_-(y,t,k)\tilde J^{(1)}(y,t,k),&k\in \R,\\
\tilde M^{(1)}(y,t,k)\rightarrow \id,&k\rightarrow \infty,
\ea
\right.
\ee
where the jump matrix $\tilde J^{(1)}(y,t,k)$ is defined by
\be\label{case1J1}
\tilde J^{(1)}(y,t,k)=\left(\ba{cc}1&\frac{r(k)}{1+|r(k)|^2}\frac{1}{\tilde \dta^2_-(k)}e^{2it\tha}\\0&1\ea\right)
\left(\ba{cc}1&0\\\frac{\ol{r(k)}}{1+|r(k)|^2}\tilde \dta^2_+(k)e^{-2it\tha}&1\ea\right)
\ee

Without loss of generality, we may assume that the function $\frac{\ol{r(k)}}{1+|r(k)|^2}e^{-2it\tha}$, i.e., the lower-triangular factor of the second matrix of the right-hand side of (\ref{case1J1}) extends analytically to the region $\im k>0$ and continuous in the closure of the region. And the first matrix of the right-hand side of (\ref{case1J1}) extends analytically to the region $\im k<0$ and continuous in the closure of the region by taking conjugate. If the analytic conditions are dropped, we have a procedure to obtain the 'weak' analytic conditions, see the Remark \ref{case1Anly}.
\par

Hence, we can introduce the following transformation if we have the proper analytic conditions:
\be
\tilde M^{(2)}(y,t,k)=\left\{
\ba{ll}
\tilde M^{(1)}(y,t,k) \left(\ba{cc}1&0\\-\frac{\ol{r(k)}}{1+|r(k)|^2}\tilde \dta^2_+(k)e^{-2it\tha}&1\ea\right),& k\in \Omega_1\cup \Omega_3,\\
\tilde M^{(1)}(y,t,k),& k\in \Omega_2\cup \Omega_5,\\
\tilde M^{(1)}(y,t,k) \left(\ba{cc}1&\frac{r(k)}{1+|r(k)|^2}\frac{1}{\tilde \dta^2_-(k)}e^{2it\tha}\\0&1\ea\right),&k\in \Omega_4\cup \Omega_6,
\ea
\right.
\ee
where $\Omega_j,j=1,2,\dots,6$ are shown in figure \ref{fig2}.
\begin{figure}[th]
\centering
\includegraphics{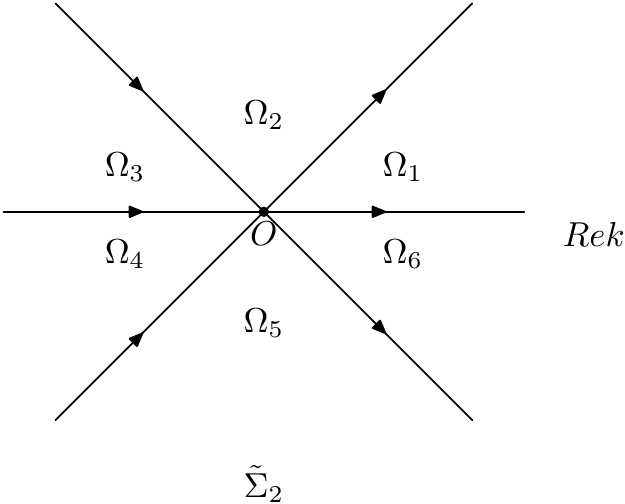}
\caption{\small The sets $\{\Omega_j\}_{j=1}^{6}$, in the $k-$plane as $\tilde\xi <0$.}\label{fig2}
\end{figure}

We obtain the new Riemann-Hilbert problem for $\tilde M^{(2)}(y,t,k)$,
\be
\left\{
\ba{l}
\tilde M_+^{(2)}(y,t,k)=\tilde M_-^{(2)}(y,t,k) \tilde J^{(2)}(y,t,k),\quad k\in \tilde\Sigma_2,\\
\tilde M^{(1)}(y,t,k)\rightarrow \id,\quad k\rightarrow \infty.
\ea
\right.
\ee
where
\be
\tilde J^{(1)}(y,t,k)=\left\{
\ba{ll}
\left(\ba{cc}1&0\\\frac{\ol{r(k)}}{1+|r(k)|^2}\tilde \dta^2_+(k)e^{-2it\tha}&1\ea\right),&k\in \tilde\Sigma_2\cap D_1,\\
\left(\ba{cc}1&\frac{r(k)}{1+|r(k)|^2}\frac{1}{\tilde \dta^2_-(k)}e^{2it\tha}\\0&1\ea\right),&k\in \tilde\Sigma_2\cap D_2.
\ea
\right.
\ee
\begin{theorem}
As $t\rightarrow \infty$, the solution $u(x,t)$ of the initial value problem (\ref{spe})-(\ref{spe-ini}) decays fast in the range $\xi>\eps$ for any $\eps >0$.
\end{theorem}
\begin{proof}
The above transformation reduces the Riemann-Hilbert problem of $\tilde M^{(2)}(y,t,k)$ to that with exponentially decaying in t to the identity matrix jump matrix. Since this Riemann-Hilbert problem is holomorphic, its solution decays fast to $\id$ and consequently $\tilde u(y,t)$ decays fast to $0$ while $y$ approaches fast $x$ and thus the domain $\tilde \xi <-\eps$ and $\xi<-\eps$ coincide asymptotically.
\end{proof}
\begin{remark}\label{case1Anly}
If the analytic conditions are dropped, then there is a procedure to obtain the 'weak' analytic conditions. To show this, we write $\frac{\ol{r(k)}}{1+|r(k)|^2}$ as a Fourier transform with respect to $\tha$,
\be
\ba{rcl}
\frac{\ol{r(k)}}{1+|r(k)|^2}e^{-2it\tha}&=&\frac{e^{-2it\tha}}{\sqrt{2\pi}(k+i)^2}\int_{-\infty}^{\infty}e^{is\tha(k)}\hat{g}(s)ds\\
{}&=&\frac{e^{-2it\tha}}{\sqrt{2\pi}(k+i)^2}\int_{t}^{\infty}e^{is\tha(k)}\hat{g}(s)ds+\frac{e^{-2it\tha}}{\sqrt{2\pi}(k+i)^2}\int_{-\infty}^{t}e^{is\tha(k)}\hat{g}(s)ds\\
{}&=&e^{-2it\tha(k)}h_{\Rmnum{1}}(k)+e^{-2it\tha(k)}h_{\Rmnum{2}}(k),
\ea
\ee
where
\[
\ba{l}
\hat{g}(s)=\frac{1}{\sqrt{2\pi}}\int_{-\infty}^{\infty}e^{-is\tha(k)}g(\tha)d\tha,\\
g(\tha)=\frac{\ol{r(k(\tha))}}{1+|r(k(\tha))|^2}(k(\tha)+i)^2.
\ea
\]

Here $e^{-2it\tha(k)}h_{\Rmnum{2}}(k)$ has an analytic continuation to the upper half-plane and decays exponentially in $L^{1}\cap L^{\infty}(\Sig\cap\{k|\im k<0\})$, as $t\rightarrow \infty$, while $e^{-2it\tha(k)}h_{\Rmnum{1}}(k)$ decays rapidly in $L^{1}\cap L^{\infty}(\R)$, as $t\rightarrow \infty$.
\end{remark}
\subsection{Proof of Theorem \ref{main2}}
$ $
\par
Denote $k_0=\frac{1}{2\sqrt{\tilde \xi}}$, then  the image part of the function $\tha$ becomes
\be
\im \tha=\frac{k_2(k_1^2+k_2^2-k_0^2)}{4k_0^2(k_1^2+k_2^2)}.
\ee
Under the assumption $\tilde \xi>\eps>0$, we have the sign picture of $\re (i\tha)$ as follows,

\begin{figure}[th]
\centering
\includegraphics{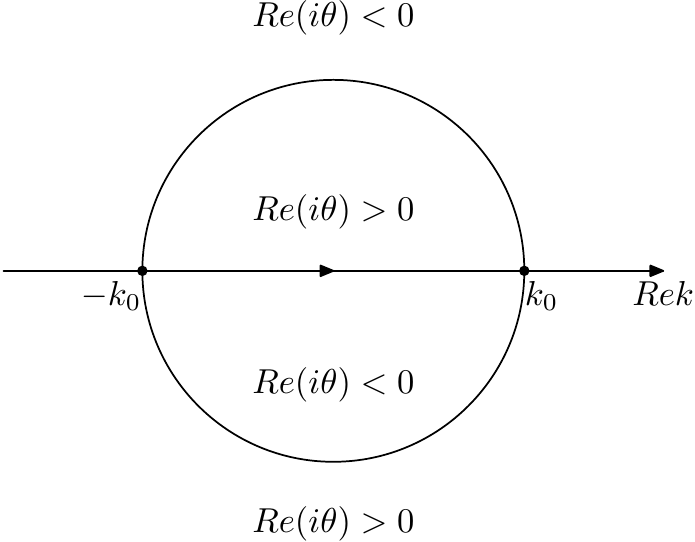}
\caption{\small The signs of $\re i\tha$ in the $k-$plane in the case $\tilde\xi >0$.}\label{fig3}
\end{figure}

Hence, this suggests the use of the following factorizations of the jump matrix $\tilde J(y,t,k)$:
{\small
\be\label{Jycase2fac}
\tilde J(y,t,k)=\left\{
\ba{ll}
\left(\ba{cc}1&0\\\ol{r(k)}e^{-2it\tha}&1\ea\right)\left(\ba{cc}1&r(k)e^{2it\tha}\\0&1\ea\right),&|k|>k_0,\\
\left(\ba{cc}1&\frac{r(k)}{1+|r(k)|^2}e^{2it\tha}\\0&1\ea\right)\left(\ba{cc}\frac{1}{1+|r(k)|^2}&0\\0&1+|r(k)|^2\ea\right)
\left(\ba{cc}1&0\\\frac{\ol{r(k)}}{1+|r(k)|^2}e^{-2it\tha}&1\ea\right),&|k|<k_0.
\ea
\right.
\ee
}

\subsubsection{The conjugate transformation}
The aim of the first transformation involves the removal of the diagonal factor in (\ref{Jycase2fac}) for $|k|<k_0$.

Introducing a scalar function $\delta(k)$ which satisfies the following scalar Riemann-Hilbert problem
\be
\left\{
\ba{rll}
\dta_+(k)&=\dta_-(k)(1+|r(k)|^2),&|k|<k_0,\\
&=\dta_-(k)=\dta(k),&|k|>k_0.\\
\dta(k)&\rightarrow 1&k\rightarrow\infty.
\ea
\right.
\ee
Then the function $\dta(k)$ is given by
\be\label{dtadef}
\dta(k)=\exp\left[\frac{1}{2\pi i}\int_{-k_0}^{k_0}\frac{\ln{(1+|r(s)|^2)}}{s-k}ds\right].
\ee
The conjugate transformation
\be\label{m1trans}
\tilde M^{(1)}(y,t,k)=\tilde M(y,t,k)\delta(k)^{\sig_3},
\ee
yields the Riemann-Hilbert problem for $\tilde M^{(1)}(y,t,k)$
\begin{subequations}\label{mu1rhp}
\be\label{Jy1rhp}
\left\{
\ba{l}
\tilde M^{(1)}_+(y,t,k)=\tilde M^{(1)}_-(y,t,k)\tilde J^{(1)}(y,t,k),\quad k\in\R,\\
\tilde M^{(1)}(y,t,k)\rightarrow \id,\quad k\rightarrow \infty,
\ea
\right.
\ee
where
\be
\tilde J^{(1)}(y,t,k)=\left\{
\ba{ll}
\left(\ba{cc}1&0\\\ol{r(k)}\dta^2e^{-2it\tha}&1\ea\right)\left(\ba{cc}1&r(k)\dta^{-2}e^{2it\tha}\\0&1\ea\right),&|k|>k_0,\\
\left(\ba{cc}1&\frac{r(k)}{1+|r(k)|^2}\dta_-^{-2}e^{2it\tha}\\0&1\ea\right)
\left(\ba{cc}1&0\\\frac{\ol{r(k)}}{1+|r(k)|^2}\dta_+^2e^{-2it\tha}&1\ea\right),&|k|<k_0.
\ea
\right.
\ee
\end{subequations}

{\bf Now, let us come back to the solution $u(x,t)$.}
From (\ref{dtadef}) it follows that
\be\label{dtalargek}
\dta(k)=\dta_0(1+k\dta_1+O(k^2)),\quad k\rightarrow 0,
\ee
where
\be\label{dta0def}
\dta_0=e^{\frac{1}{2\pi i}\int_{-k_0}^{k_0}\frac{\ln(1+|r(s)|^2)}{s}ds},\quad \dta_1=\frac{1}{2\pi i}\int_{-k_0}^{k_0}\frac{\ln{(1+|r(s)|^2)}}{s^2}ds.
\ee
If we write
\be
\tilde M(y,t,k)=\tilde M_0(y,t)+k\tilde M_1(y,t)+O(k^2),\quad k\rightarrow 0,
\ee
and
\be
\tilde M^{(1)}(y,t,k)=\tilde M^{(1)}_0(y,t)+k\tilde M^{(1)}_1(y,t)+O(k^2),\quad k\rightarrow 0,
\ee
then from the transformation (\ref{m1trans}) we obtain
\be
\tilde M_0(y,t)=\tilde M^{(1)}_0(y,t)\dta_0^{-\sig_3},\quad \tilde M_1(y,t)=(\tilde M^{(1)}_1(y,t)-\dta_1\tilde M^{(1)}_0(y,t)\sig_3)\dta_0^{-\sig_3},
\ee
Hence, we have
\begin{subequations}
\be
iu(x,t)e^{-2d}=\dta_0^2\left[(\tilde M^{(1)}_0)^{-1}\tilde M^{(1)}_1\right]_{12},
\ee
\be
c_+=-i\left(\left[(\tilde M^{(1)}_0)^{-1}\tilde M^{(1)}_1\right]_{11}-\dta_1\right).
\ee
\end{subequations}

\subsubsection{The ``open lense" transformation}
For the convenience of the notation, we transverse the direction of the component $|k|>k_0$ of the jump contour $\R$ for the Riemann-Hilbert problem of $\tilde M^{(1)}(y,t,k)$. Like the proof of Theorem \ref{main1}, we assume that the functions appeared in the right-hand side of the jump matrix $\tilde J^{(1)}(y,t,k)$ can be analytic extension to the proper regions. If the analytic extension dropped, then we can do a similar procedure to obtain the 'weak' analytic condition as the appendix A of \cite{jiansp}.

If we write
\be
\tilde J^{(1)}(y,t,k)=(b_-^{(1)}(y,t,k))^{-1}b_+^{(1)}(y,t,k),
\ee
then, we can make a transformation as
\be
\tilde M^{(2)}(y,t,k)=\tilde M^{(1)}(y,t,k)T^{(12)}(y,t,k),
\ee
where $T^{(12)}(y,t,k)$ is defined as
\be
T^{(12)}(y,t,k)=\left\{
\ba{ll}
(b_-^{(1)}(y,t,k))^{-1},&k\in \Omega_1\cup\Omega_3\cup\Omega_9\cup\Omega_{10},\\
\id,&k\in \Omega_2\cup\Omega_5,\\
(b_+^{(1)}(y,t,k))^{-1},&k\in \Omega_4\cup\Omega_6\cup\Omega_7\cup\Omega_{8},
\ea
\right.
\ee
with the $\Omega_j$ showed as the Figure \ref{fig4}.

\begin{figure}[th]
\centering
\includegraphics{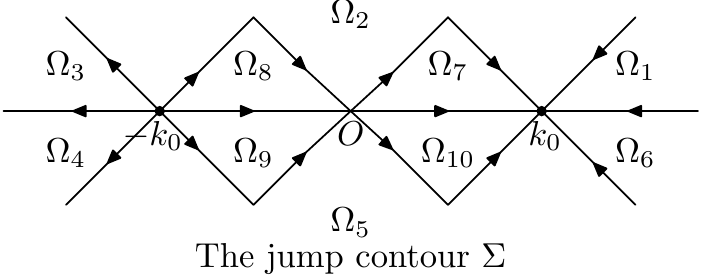}
\caption{\small The sets $\{\Omega_j\}_{j=1}^{10}$, in the $k-$plane as $\tilde\xi >0$.}\label{fig4}
\end{figure}

Thus, we obtain the Riemann-Hilbert problem for $\tilde M^{(2)}(y,t,k)$
\begin{subequations}\label{m2rhp}
\be
\left\{
\ba{l}
\tilde M^{(2)}_+(y,t,k)=\tilde M^{(2)}_-(y,t,k)\tilde J^{(2)}(y,t,k),\quad k\in \Sigma\\
\tilde M^{(2)}(y,t,k)\rightarrow \id,\quad k\rightarrow \infty.
\ea
\right.
\ee
where
\be
\ba{rcl}
\tilde J^{(2)}(y,t,k)&=&
(b^{(2)}_-(y,t,k))^{-1}b^{(2)}_+(y,t,k)\\
{}&=&\left\{
\ba{ll}
\id,&k\in \R,\\
b_+^{(1)}(y,t,k),&k\in \Sigma\cap\re i\tha>0,\\
(b_-^{(1)}(y,t,k))^{-1},&k\in \Sigma\cap\re i\tha<0.
\ea
\right.
\ea
\ee
\end{subequations}

 Now let us come back to the solution $u(x,t)$ again. The solution $u(y,t)$ is related to the solution of the Riemann-Hilbert problem evaluated at $k=0$, it may be affected by this transformation. However, due to the following proposition, this transformation turns out not to affect the terms in the expansion of the solution of the Riemann-Hilbert problem at $k=0$ at least up to the terms of order $O(k^2)$ and thus it does not really affect the leading order asymptotic behavior of $u(y,t)$.
\begin{proposition}
The reflection coefficient $r(k)=O(k^3)$ as $k\rightarrow 0$.
\end{proposition}
\begin{proof}
A direct calculation following from (\ref{akasy}) and from the identity $|r(k)|^2=\frac{1}{|a(k)|^2}-1$.
\end{proof}

So, if we write
\be
\tilde M^{(2)}(y,t,k)=\tilde M^{(2)}_0(y,t)+k\tilde M^{(2)}_1(y,t)+O(k^2),\quad k\rightarrow 0,
\ee
then we have
\begin{subequations}
\be
iu(x,t)e^{-2d}=\dta_0^2\left[(\tilde M^{(2)}_0)^{-1}\tilde M^{(2)}_1\right]_{12},
\ee
\be
c_+=-i\left(\left[(\tilde M^{(2)}_0)^{-1}\tilde M^{(2)}_1\right]_{11}-\dta_1\right).
\ee
\end{subequations}

Set
\be
\omega_{\pm}^{(2)}(y,t,k)=\pm(b^{(2)}_{\pm}(y,t,k)-\id),\quad \omega=\omega^{(2)}_++\omega^{(2)}_-,
\ee
and let $\mu^{(2)}(y,t,k)$ be the solution of the singular integral equation $\mu^{(2)}=\id+C_{\omega}\mu^{(2)}$, here $C_{\omega}$ is defined as
\[
C_{\om}f=C_+(f\om_-)+C_-(f\om_+),\quad \forall f\mbox{ is a $2\times 2$ matrix-valued function},
\]
where
\[
(C_{\pm}f)(k)=\int_{\Gam}\frac{f(\xi)}{\xi-k_{\pm}}\frac{d\xi}{2\pi i},\quad k\in\Gam,f\in L^2(\Gam),
\]
then
\be\label{m2sol}
\tilde M^{(2)}(y,t,k)=\id+\frac{1}{2\pi i}\int_{\Sigma}\frac{\mu^{(2)}(y,t,\eta)\omega(y,t,\eta)}{\eta-k}d\eta,\quad k\in \C\backslash \Sigma
\ee
is the solution of Riemann-Hilbert problem (\ref{m2rhp}).
\par
Expanding the integral (\ref{m2sol}) around $k=0$, we have
\begin{subequations}\label{m20and1}
\be
\tilde M^{(2)}_0(y,t)=\id+\frac{1}{2\pi i}\int_{\Sigma}\frac{\mu^{(2)}(y,t,\eta)\omega(y,t,\eta)}{\eta}d\eta,
\ee
\be
\tilde M^{(2)}_1(y,t)=\frac{1}{2\pi i}\int_{\Sigma}\frac{\mu^{(2)}(y,t,\eta)\omega(y,t,\eta)}{\eta^2}d\eta.
\ee
\end{subequations}
\begin{remark}
Since, $\omega(y,t,k)$ decays rapidly at $0$, the integral (\ref{m20and1}) are nonsingular.
\end{remark}

\subsubsection{Reduction to the model problem}
Following a similar procedure as section 5.4 and 5.5 in \cite{jiansp}, as $t\rightarrow \infty$, we can reduce the Riemann-Hilbert problem for $\tilde M^{(2)}(y,t,k)$ to two small "crosses" which are centered around $\pm k_0$, showed in Figure \ref{fig5}. And the contributions of these two crosses to the leading order asymptotic behavior of the solution $u(x,t)$ are separated out.
\begin{figure}[th]
\centering
\includegraphics{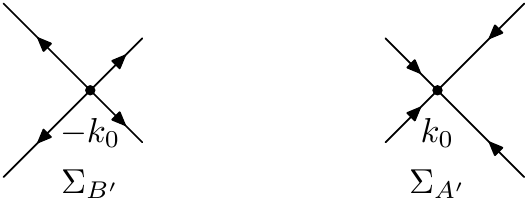}
\caption{\small The jump contour $\Sigma'$ in the $k-$plane as $\tilde\xi >0$.}\label{fig5}
\end{figure}

In order to calculate the explicit asymptotic behavior of the solution $u(x,t)$ in terms of the Riemann-Hilbert problem for $\tilde M^{(2)}(y,t,k)$, as $t\rightarrow \infty$, we need reduce the Riemann-Hilbert problem to a model problem whose solution can be given explicitly in terms of parabolic cylinder functions. To do this, we should evaluate the leading term of the function $\dta(k)e^{-it\tha(k)}$ as $k\rightarrow \pm k_0$.
\par
Recall that
\be
\dta(k)=e^{\frac{1}{2\pi i}\int_{-k_0}^{k_0}\frac{\ln(1+|r(s)|^2)}{s-k}ds}=\frac{(k-k_0)^{i\nu(k_0)}}{(k+k_0)^{i\nu(-k_0)}}e^{\chi(k)},
\ee
where
\be\label{nudef}
\nu(k)=-\frac{1}{2\pi}\ln(1+|r(k)|^2),\quad \chi(k)=-\frac{1}{2\pi i}\int_{-k_0}^{k_0}\ln(k-s)d(\ln(1+|r(s)|^2)).
\ee
And as $k\rightarrow k_0$,
\be
\tha(k)=\frac{1}{2k_0}+\frac{1}{4k_0^3}(k-k_0)^2-\frac{1}{4\eta^4}(k-k_0)^3,\quad \mbox{$\eta$ lies between $k$ and $k_0$}.
\ee
\par
Let us consider a scaling transformation as
\be
(N_{(k_0)}f)(k)=f(k_0+\sqrt{\frac{k_0^3}{t}}k),
\ee
then the function $\dta(k)e^{-it\tha(k)}$ becomes
\be
(N_{(k_0)}\dta e^{-it\tha})(k)=\dta^{(0)}_{(k_0)}\dta^{(1)}_{(k_0)},
\ee
where
\begin{subequations}
\be
\dta^{(0)}_{(k_0)}=(\frac{k_0^3}{t})^{\frac{i\nu(k_0)}{2}}(2k_0)^{-i\nu(-k_0)}e^{\chi(k_0)}e^{-\frac{it}{2k_0}},
\ee
\be
\dta^{(1)}_{(k_0)}=k^{i\nu(k_0)}(\frac{2k_0+\sqrt{\frac{k_0^3}{t}}k}{2k_0})^{-i\nu(-k_0)}e^{\chi(k_0+\sqrt{\frac{k_0^3}{t}}k)-\chi(k_0)}e^{-\frac{ik^2}{4}},
\ee
\end{subequations}
here  $k^{i\nu(k_0)}$ is cut along $(-\infty,0)$.
\par
For $k\rightarrow -k_0$, the scaling transformation is
\be
(N_{(-k_0)}f)(k)=f(-k_0+\sqrt{\frac{k_0^3}{t}}k),
\ee
then the function $\dta(k)e^{-it\tha(k)}$ becomes
\be
(N_{(-k_0)}\dta e^{-it\tha})(k)=\dta^{(0)}_{(-k_0)}\dta^{(1)}_{(-k_0)},
\ee
where
\begin{subequations}
\be
\dta^{(0)}_{(-k_0)}=(\frac{k_0^3}{t})^{\frac{-i\nu(-k_0)}{2}}(2k_0)^{i\nu(k_0)}e^{\tilde \chi(-k_0)}e^{\frac{it}{2k_0}},
\ee
\be
\dta^{(1)}_{(-k_0)}=(-k)^{-i\nu(-k_0)}(\frac{2k_0-\sqrt{\frac{k_0^3}{t}}k}{2k_0})^{i\nu(k_0)}e^{\tilde \chi(-k_0+\sqrt{\frac{k_0^3}{t}}k)-\tilde \chi(-k_0)}e^{\frac{ik^2}{4}}
\ee
\end{subequations}
with
\be
\tilde \chi(k)=-\frac{1}{2\pi i}\int_{-k_0}^{k_0}\ln(s-k)d\ln(1+|r(s)|^2),
\ee
here $(-k)^{-i\nu(-k_0)}$ is cut along $(0,\infty)$.
\par
Hence, following \cite{jiansp}, the solution of the Riemann-Hilbert problem for $\tilde M^{(2)}(y,t,k)$ formulated on two separated crosses centered at $k=\pm k_0$ can be approximated, for large t, in terms of the solution $\tilde M^{(model)}(y,t,k)$ of the model problem formulated on a cross centered at $k=0$ and evaluated for large $k$.
\par
For the $k$ centered at $k_0$ case, the solution of the model Riemann-Hilbert problem for $\tilde M^{(model)}(y,t,k)$ has the following form,
as $\im k>0$,
\be
\tilde M^{(model)}(y,t,k)=\left(\ba{cc}
\tilde M^{(model)+}_{11}&\tilde M^{(model)+}_{12}\\
\tilde M^{(model)+}_{21}&\tilde M^{(model)+}_{22}
\ea
\right),
\ee
where
\be
\ba{l}
\tilde M^{(model)+}_{11}=e^{\frac{\pi}{4}\nu(k_0)}D_{-a_1}(e^{-\frac{i\pi}{4}}k),\\
\tilde M^{(model)+}_{21}=\frac{1}{\beta^{k_0}_{12}}e^{\frac{\pi}{4}\nu(k_0)}[\frac{d}{dk}D_{-a_1}(e^{-\frac{i\pi}{4}}k)-\frac{ik}{2}D_{-a_1}(e^{-\frac{i\pi}{4}}k)],\\
\tilde M^{(model)+}_{22}=e^{-\frac{3\pi}{4}\nu(k_0)}D_{a_1}(e^{-\frac{3i\pi}{4}}k),\\
\tilde M^{(model)+}_{12}=\frac{1}{\beta^{k_0}_{21}}e^{-\frac{3\pi}{4}\nu(k_0)}[\frac{d}{dk}D_{a_1}(e^{-\frac{3i\pi}{4}}k)+\frac{ik}{2}D_{a_1}(e^{-\frac{3i\pi}{4}}k)]
\ea
\ee
with $a_1=i\nu(k_0)$,
and $\beta^{k_0}_{12}=i\tilde M^{(k_0)1}_{12}$, $\beta^{k_0}_{21}=-i\tilde M^{(k_0)1}_{21}$.
Here
$\tilde M^{(k_0)1}_{ij}$ denotes the ($i,j$)-th element of the $2\times 2$ matrix $\tilde M^{(k_0)1}$ defined by (\ref{tildemk01def}),
\be\label{tildemk01def}
\tilde M_{k_0}(y,t,k)=\id-\frac{\tilde M^{(k_0)1}}{k}+O(\frac{1}{k^2}),\quad k\rightarrow \infty,
\ee
where $\tilde M_{k_0}(y,t,k)$ solves the following the Riemann-Hilbert problem
\be
\left\{
\ba{ll}
\tilde M_{k_0+}(y,t,k)=\tilde M_{k_0-}(y,t,k)\tilde J^{(k_0)}(x,t,k),&k\in \R,\\
\tilde M_{k_0}(y,t,k)\rightarrow \id,&k\rightarrow \infty,
\ea
\right.
\ee
with
\be
\tilde J^{(k_0)}(x,t,k)=k^{-i\nu(k_0)\hat\sig_3}e^{\frac{ik^2}{4}\hat\sig_3}\left(\ba{cc}1&r(k_0)\\ \ol{r(k_0)}&1+|r(k_0)|^2\ea\right).
\ee

Similarly, as $\im k<0$,
\be
\ba{l}
\tilde M^{(model)-}_{11}=e^{-\frac{3\pi}{4}\nu(k_0)}D_{-a_1}(e^{\frac{3i\pi}{4}}k)\\
\tilde M^{(model)-}_{21}=\frac{1}{\beta^{k_0}_{12}}e^{-\frac{3\pi}{4}\nu(k_0)}[\frac{d}{dk}D_{-a_1}(e^{\frac{3i\pi}{4}}k)-\frac{ik}{2}D_{-a_1}(e^{\frac{3i\pi}{4}}k)]\\
\tilde M^{(model)-}_{22}=e^{\frac{\pi}{4}\nu(k_0)}D_{a_1}(e^{\frac{i\pi}{4}}k)\\
\tilde M^{(model)-}_{12}=\frac{1}{\beta^{k_0}_{21}}e^{\frac{\pi}{4}\nu(k_0)}[\frac{d}{dk}D_{a_1}(e^{\frac{i\pi}{4}}k)+\frac{ik}{2}D_{a_1}(e^{\frac{i\pi}{4}}k)].
\ea
\ee

Hence, from
\be
(\tilde M^{(model)}_{-}(y,t,k))^{-1}(\tilde M^{(model)}_{+}(y,t,k))=\left(\ba{cc}1&r(k_0)\\ \ol{r(k_0)}&1+|r(k_0)|^2\ea\right),
\ee
we have
\be
\ba{rcl}
r(k_0)&=&\displaystyle{\frac{e^{-\frac{\pi}{2}\nu(k_0)}}{\beta^{k_0}_{21}}}\mbox{Wr}[D_{a_1}(e^{\frac{i\pi}{4}}k),D_{a_1}(e^{-\frac{3i\pi}{4}}k)]\\
{}&=&\displaystyle{\frac{e^{-\frac{\pi}{2}\nu(k_0)}}{\beta^{k_0}_{21}}\frac{\sqrt{2\pi}e^{\frac{i\pi}{4}}}{\Gam(-a_1)}}.
\ea
\ee
Here $\mbox{Wr}[D_{a_1}(e^{\frac{i\pi}{4}}k),D_{a_1}(e^{-\frac{3i\pi}{4}}k)]$ denotes the Wronskian  determinant.
By $\beta^{k_0}_{12}\beta^{k_0}_{21}=\nu(k_0)$, we can compute
\be
\beta^{k_0}_{12}=-\frac{e^{-\frac{\pi}{2}\nu(k_0)}}{\ol{r(k_0)}}\frac{\sqrt{2\pi}e^{-\frac{i\pi}{4}}}{\Gam(a_1)}.
\ee

Similarly, we can show that the solution of the model problem corresponding to $-k_0$ as follows,
\par
\begin{subequations}
for $k\in \im k>0$,
\be
\ba{l}
\tilde M^{(model)+}_{11}=e^{-\frac{3\pi}{4}\nu(-k_0)}D_{a_2}(e^{-\frac{3i\pi}{4}}k)\\
\tilde M^{(model)+}_{21}=\frac{1}{\beta^{-k_0}_{12}}e^{-\frac{3\pi}{4}\nu(-k_0)}[\frac{d}{dk}D_{a_2}(e^{-\frac{3i\pi}{4}}k)+\frac{ik}{2}D_{a_2}(e^{-\frac{3i\pi}{4}}k)]\\
\tilde M^{(model)+}_{22}=e^{\frac{\pi}{4}\nu(-k_0)}D_{-a_2}(e^{-\frac{i\pi}{4}}k)\\
\tilde M^{(model)+}_{12}=\frac{1}{\beta^{-k_0}_{21}}e^{\frac{\pi}{4}\nu(-k_0)}[\frac{d}{dk}D_{-a_2}(e^{-\frac{i\pi}{4}}k)-\frac{ik}{2}D_{-a_2}(e^{-\frac{i\pi}{4}}k)];
\ea
\ee
for $k\in \im k<0$,
\be
\ba{l}
\tilde M^{(model)-}_{11}=e^{\frac{\pi}{4}\nu(-k_0)}D_{a_2}(e^{\frac{i\pi}{4}}k)\\
\tilde M^{(model)-}_{21}=\frac{1}{\beta^{-k_0}_{12}}e^{\frac{\pi}{4}\nu(-k_0)}[\frac{d}{dk}D_{a_2}(e^{\frac{i\pi}{4}}k)+\frac{ik}{2}D_{a_2}(e^{\frac{i\pi}{4}}k)]\\
\tilde M^{(model)-}_{22}=e^{-\frac{3\pi}{4}\nu(-k_0)}D_{-a_2}(e^{\frac{3i\pi}{4}}k)\\
\tilde M^{(model)-}_{12}=\frac{1}{\beta^{-k_0}_{21}}e^{-\frac{3\pi}{4}\nu(-k_0)}[\frac{d}{dk}D_{-a_2}(e^{3\frac{i\pi}{4}}k)-\frac{ik}{2}D_{-a_2}(e^{\frac{3i\pi}{4}}k)],
\ea
\ee
\end{subequations}
Then, we have
\be
\beta^{-k_0}_{12}=\frac{e^{-\frac{\pi}{2}\nu(-k_0)}}{\ol{r(-k_0)}}\frac{\sqrt{2\pi}e^{\frac{i\pi}{4}}}{\Gam(-a_2)},\quad
\beta^{-k_0}_{21}=-\frac{e^{-\frac{\pi}{2}\nu(-k_0)}}{r(-k_0)}\frac{\sqrt{2\pi}e^{-\frac{i\pi}{4}}}{\Gam(a_2)},
\ee
with $a_2=i\nu(-k_0)$, $\beta^{-k_0}_{12}=-i\tilde M^{(-k_0)1}_{12}$ and $\beta^{-k_0}_{21}=i\tilde M^{(-k_0)1}_{21}$. The definition of $\tilde M^{(-k_0)1}_{ij}$ is similar as $\tilde M^{(k_0)1}_{ij}$ (see equation (\ref{tildemk01def})).

\par
Then, similarly as the formulae (5.110) in \cite{jiansp}, we can know the leading order asymptotic behavior of the solution $u(y,t)$ as $t\rightarrow \infty$ comes from the $\tilde M^{(k_0)1}_{12}$ and $\tilde M^{(-k_0)1}_{12}$, that is,
 \be
 iu(y,t)e^{-2d}=\frac{1}{\sqrt{k_0t}}\left[\tilde M^{(k_0)1}_{12}+\tilde M^{(-k_0)1}_{12}\right].
 \ee
A direct calculation shows that the leading order asymptotic behavior of the solution $u(y,t)$ is
\be
iu(y,t)e^{-2d}=\dta_0^2\frac{1}{\sqrt{t}}\left[
\sqrt{\frac{-\nu(-k_0)}{k_0}}e^{i\phi_1}-\sqrt{\frac{-\nu(k_0)}{k_0}}e^{-i\phi_2}
\right],
\ee
where
\begin{subequations}
\be
\ba{l}
\phi_1=\frac{i\pi}{4}+\mbox{arg}(r(-k_0))+\mbox{arg}(\Gam(i\nu(-k_0)))+\nu(-k_0)\ln(\frac{k_0^3}{t})\\
{}-\nu(k_0)\ln(4k_0^2)-\frac{1}{\pi}\int_{-k_0}^{k_0}\ln(s+k_0)d\ln(1+|r(s)|^2)-\frac{t}{k_0},
\ea
\ee
\be
\ba{l}
\phi_2=\frac{i\pi}{4}-\mbox{arg}(r(-k_0))-\mbox{arg}(\Gam(-i\nu(k_0)))+\nu(k_0)\ln(\frac{k_0^3}{t})\\
{}-\nu(-k_0)\ln(4k_0^2)+\frac{1}{\pi}\int_{-k_0}^{k_0}\ln(k_0-s)d\ln(1+|r(s)|^2)-\frac{t}{k_0}.
\ea
\ee
\end{subequations}

The relation $c_+$ lies between the new scale $y$ and the original scale $x$,
\be
ic_+=-\dta_1+o(1),\quad t\rightarrow \infty.
\ee
Hence, the solution $u(x,t)$ has the asymptotic (\ref{uxtasy-final}).
%
%

\bigskip

{\bf Acknowledgements}
This work was supported by National Science Foundation of China under project No. 11671095.
Xu was also supported by National Science Foundation of China under project No. 11501365, Shanghai Sailing Program
supported by Science and Technology Commission of Shanghai Municipality under Grant No. 15YF1408100, Shanghai youth teacher assistance program No. ZZslg15056.
And Xu wants to give many thanks to Shanghai Center Mathematical Science, many of this work was done during Xu visited there.

\end{document}